\documentclass[12pt]{article}
\usepackage{amsfonts}
\usepackage{amsmath, amsthm, upref}
\usepackage{enumerate}
\usepackage{color}
\usepackage{xcolor}
\usepackage{hyperref}
\usepackage{multimedia}
\usepackage{graphics}
\usepackage{epsfig}
\usepackage{graphicx}
\usepackage[english]{babel}
\usepackage{textcomp}
\usepackage{eurosym}
\usepackage{amssymb}

\newcommand{\comment}[1]{}

\newtheorem{theorem}{Theorem}

\newtheorem{assumption}[theorem]{Assumption}

\newtheorem{lemma}{Lemma}

\newcommand{\bee}{\begin{equation}}
\newcommand{\eee}{\end{equation}}
\newcommand{\bea}{\begin{eqnarray}}
\newcommand{\eea}{\end{eqnarray}}
\newcommand{\bean}{\begin{eqnarray*}}
\newcommand{\eean}{\end{eqnarray*}}

\begin{document}
\nocite{*}

\title{ From quadratic Hawkes processes to super-Heston rough volatility models with Zumbach effect}
\author{Aditi Dandapani, Paul Jusselin and Mathieu Rosenbaum\\
$~~$\\
\'Ecole Polytechnique}

\date{\today }
\maketitle

\begin{abstract}
Using microscopic price models based on Hawkes processes, it has been shown that under some no-arbitrage condition, the high degree of endogeneity of markets together with the phenomenon of metaorders splitting generate rough Heston-type volatility at the macroscopic scale. One additional important feature of financial dynamics, at the heart of several influential works in econophysics, is the so-called feedback or Zumbach effect. This essentially means that past trends in returns convey significant information on future volatility. A natural way to reproduce this property in microstructure modeling is to use quadratic versions of Hawkes processes.  We show that after suitable rescaling, the long term limits of these processes are refined versions of rough Heston models where the volatility coefficient is enhanced compared to the square root characterizing Heston-type dynamics. Furthermore the Zumbach effect remains explicit in these limiting rough volatility models.
\end{abstract}

\section{Introduction}

\label{intro}
Since the paper \cite{gatheral2018volatility}, it has been well accepted that volatility is rough. This means that log-volatility essentially behaves as fractional Brownian motion with Hurst parameter of order $0.1$, see also for example \cite{bennedsen2016decoupling,da2018volatility,glasserman2018buy,livieri2018rough}. There are microstructural foundations for rough volatility that use Hawkes processes to create a microscopic model for asset prices. In this vein, the authors in \cite{el2018microstructural} consider four stylized facts concerning market microstructure: the high degree of endogeneity of markets, the no-arbitrage property, buying/selling asymmetry and the long memory of the market order flow generated by metaorders. They show that when only the three first stylized facts are taken into account, one obtains the Heston model for the scaling limit of the price process. When the long memory property of the flow is added, the limit is the rough Heston model introduced and developed in \cite{el2018perfect,el2019characteristic}. In the rough Heston model, the spot variance 
$V_t$ can be written as follows:
\begin{equation}
\label{eq:rough_heston}
V_t = V_0+ \frac{\lambda}{\Gamma(1-\alpha)}\int_0^t (t-s)^{\alpha - 1}\big(\theta_0(s) - V_s \big)\mathrm{d}s + \nu\sqrt{V_s}\mathrm{d}B_s,
\end{equation}
where $\lambda$ and $\nu$ are some positive constants, $\theta_0$ is a deterministic function, $\alpha \in (1/2, 1)$ and $B$ is a Brownian motion. The rough behavior is due to the singular kernel $(t-s)^{\alpha - 1}$ which is the same as that appearing in the Mandelbrot-van Ness representation of a fractional Brownian motion with Hurst parameter $\alpha-1/2$. More recently, assuming only that the order flow is driven by a linear Hawkes process and that there is no statistical arbitrage on the market, it is shown in  \cite{jusselin2018no} that the price necessarily follows a rough Heston model. In fact, as far as we know, all the works on microstructural foundations of rough volatility have hitherto produced a rough Heston model.\\

\noindent However, in the context of rough models, there are other aspects of volatility that one could wish to understand from a microstructural perspective. A first point is to go beyond the square root associated to the dynamic of the volatility in the rough Heston model \eqref{eq:rough_heston}. A particularly interesting case is when an additional additive or multiplicative term appears, enhancing the square root and leading to fatter volatility tails, see \cite{jaber2018multi,blanc2017quadratic}.\\ 

\noindent Another important stylized fact of financial time series is the feedback of price returns on volatility. This phenomenon is introduced by Zumbach in \cite{zumbach2010volatility} where he measures the impact of price trends on future volatility, see also \cite{lynch2003market,zumbach2009time}. It is demonstrated that price trends induce an increase of volatility. We refer to this property as {\it Zumbach effect}. In the literature, see notably \cite{chicheportiche2014fine}, a way to reinterpret the Zumbach effect is to consider that the predictive power of past squared returns on future volatility is stronger than that of past volatility on future squared returns. To check this on data, one typically shows that the covariance between past squared price returns and future realized volatility (over a given duration) is larger than that between past realized volatility and future squared price returns, see \cite{blanc2017quadratic,chicheportiche2014fine,el2020zumbach} for more details. We refer to this version of Zumbach effect as {\it weak Zumbach effect}.\\

\noindent It has been proved in \cite{el2020zumbach} that the rough Heston model reproduces the weak form of Zumbach effect. However, it is not obtained through feedback effect, which is the motivating phenomenon in the original paper by Zumbach \cite{zumbach2010volatility}. It is only due to the dependence between price and volatility created by the correlation of the Brownian motions driving their dynamics. In particular in the rough Heston model, the conditional law of the volatility depends on the past dynamic of the price only through the past volatility, see \cite{el2018perfect}. From now on, we speak about {\it strong Zumbach effect} when the conditional law of future volatility depends not only on past volatility trajectory but also on past returns.\\

\noindent Inspired by the methodology of \cite{el2018microstructural}, our goal in this paper is to propose microstructural foundations for the strong Zumbach effect. We also wish to obtain models such that the instantaneous volatility of variance is equal to the classical square root term of Heston like models multiplied by a non-trivial process, in order to enhance volatility tails. Any model satisfying the latter property will be called a {\it super-Heston model}.\\

\noindent A convenient way to build a microscopic model, encoding Zumbach effect and leading naturally to super-Heston rough volatility, is to use a quadratic Hawkes based price process, in the spirit of \cite{blanc2017quadratic}. More precisely, we consider the following microstructural model for the price $(P_t)_{t\geq 0}$: it is piecewise constant with sizes of price jumps independent and identically distributed taking values $1$ or $-1$ with probability $1/2$. The jump times are those of a counting process $N$. We assume that $N$ is a quadratic Hawkes type process as introduced in \cite{blanc2017quadratic, ogata1981lewis}. This means the intensity $(\lambda_t)_{t\geq 0}$ of $N$ is given by 
\begin{equation}
\label{eq:intensity_quadratic_hawkes}
\lambda_t = \mu + \int_{0}^t \phi(t-s) \mathrm{d}N_s + Z_t^2,\text{ with }Z_t =  \int_0^t  k(t-s)\mathrm{d}P_s,
\end{equation}
where $\phi$ and $k$ are two non-negative measurable functions supported on $\mathbb{R}_+$ and $\mu>0$. In the definition of the intensity, the linear term with kernel $\phi$ enables us to model the self-exciting nature of the order flow. The component $Z_t$ is a moving average of past returns. It can be thought of as a proxy for price return over a given time horizon. If the price has been essentially trending in the past, $Z_t$ is large leading to high intensity. On the contrary if it has been oscillating, $Z_t$ is close to zero and there is low feedback from the returns on the intensity. Hence $Z_t$ can obviously be understood as a (strong) Zumbach term. Note that of course one can think that positive and negative price trends have different impact on the volatility. However for simplicity we neglect this asymmetry in this paper. Finally recall that the stability condition for Model \eqref{eq:intensity_quadratic_hawkes} is $\|\phi \|_1 + \|k\|_2^2$ being strictly smaller than one, see \cite{blanc2017quadratic}.\\

\noindent Remark that if we forget the quadratic term $Z_t^2$ in the intensity, we are left with a linear Hawkes process just as in \cite{Jaisson:2016aa}. In this case, at the scaling limit, if the kernel $\phi$ is heavy tailed and if we are near instability, meaning $\|\phi\|_1$ tends to one with the time parameter driving the asymptotic, the rescaled intensity process converges in law to a rough dynamic similar to \eqref{eq:rough_heston}, see \cite{el2018microstructural,Jaisson:2016aa}. When the kernel norm $\| \phi \|_1$ is fixed and strictly smaller than one, a deterministic limiting model is obtained. Thus we see that being in the near instability regime is crucial so that roughness can arise from the kernel $\phi$. Recall that this regime corresponds to a high degree of endogeneity of the market, see \cite{filimonov2012quantifying,hardiman2013critical,Jaisson:2015aa,Jaisson:2016aa}\\

\noindent In \cite{blanc2017quadratic}, the authors study the long term behavior of the intensity of quadratic Hawkes processes. That is, on the time horizon $[0,T],$ letting $T$ tend to infinity, they are interested in the limiting dynamic of $(\lambda_{tT})_{t\in[0,1]}$, which can be viewed as the macroscopic (squared) volatility. They work in a setting where $\|\phi \|_1 + \|k\|_2^2=2\gamma< 1$ is fixed, not depending on $T$. Based on PDE techniques, they obtain a diffusion process with power-law marginal distributions and strong Zumbach effect for the asymptotic volatility. More precisely, their limiting model $(\hat{P}_t, V_t)$ for price and volatility writes as follows:
$\mathrm{d}\hat{P}_t = \sqrt{V_t}\mathrm{d}B_t$ with
\begin{align*}
V_t &= \mu + (Z_t)^2 + \int_0^t \gamma \beta e^{-\beta(t-s)} V_s\mathrm{d}s\\
Z_t &= \int_0^t \sqrt{\gamma\alpha } e^{-(t-s)\alpha/2}\mathrm{d}\hat{P}_s,
\end{align*}
with $B$ a Brownian motion and $\alpha$, $\beta$ some positive parameters defining the functions $\phi$ and $k$ taken exponential in \cite{blanc2017quadratic}.\\

\noindent In this paper, we wish to go beyond the case treated in \cite{blanc2017quadratic} from which we draw inspiration.  
We describe further relevant limiting price dynamics that can be generated from quadratic Hawkes processes. We focus on finding microscopic basis for super-Heston rough volatility processes with strong Zumbach effect. Our goal is to establish connections between micro-parameters of the quadratic Hawkes dynamic and macro-phenomena such as the roughness of the volatility and the strong Zumbach effect.\\

\noindent We first focus in Section \ref{sec:purely_quadratic} on the purely quadratic case, that is when $\phi$ is equal to zero. Choosing appropriate scaling parameters, we obtain the following limiting model:
$\mathrm{d}\hat{P}_t = \sqrt{V_t}\mathrm{d}B_t$ with
\begin{align}\label{eq:purely_quadratic_model}
V_t &= \mu + Z_t^2\\
Z_t &=\sqrt{\gamma}\int_0^t k(t-s)\mathrm{d}\hat{P}_s,\nonumber
\end{align}
where $\gamma\in(0,1)$ is related to the scaling of the kernel $k$.
In contrast to the purely linear case, we do not need any sort of near instability here so that a stochastic volatility model arises at the scaling limit. In \eqref{eq:purely_quadratic_model} the strong Zumbach effect is naturally encoded since the volatility is a functional of past price returns through $Z$. We also have that the quadratic feedback of price returns on volatility implies that $V_t$ is of super-Heston type (essentially log-normal here). This can be seen for instance when $\mu=0$ where we get
$$Z_t = \sqrt{\gamma}\int_0^t k(t-s)|Z_s|\mathrm{d}B_s.$$
Moreover taking for example $k= f^{H + 1/2, \lambda}$ for $H\in (0, 1/2)$ and $\lambda>0$ with $f^{\alpha, \lambda}$ the Mittag-Leffler function\footnote{See \cite{el2019characteristic} for a reminder and connections with the Mandelbrot-van Ness representation of fractional Brownian motion.}, we get that the volatility has H\"older regularity $H-\varepsilon$ for any $\varepsilon>0$. Thus, from a natural microscopic dynamic, we are able to obtain a super-Heston rough volatility model with strong Zumbach effect at the macroscopic limit.\\

\noindent We then investigate the limiting models arising from quadratic Hawkes processes with non-vanishing linear part. Knowing that roughness can be obtained from the linear part only in the near instability regime, we treat separately this case and the stable one. We consider in Section \ref{sec:full_quadratic_stable} the situation where the stability condition is not asymptotically violated. The result is similar to \eqref{eq:purely_quadratic_model} up to the addition of a drift term $\beta\int_0^t\phi(t-s)V_s\mathrm{d}s$ in the dynamic of $V_t$, where $\beta$ is a constant related to the scaling procedure.\\

\noindent We study the nearly unstable case where the $L^1$\footnote{We use the notation $L^p$ without reference to the underlying domain when no confusion is possible.} norm of the kernel driving the linear part tends to one in Section \ref{sec:unstable_case}. Assuming $\phi(x)$ behaves as $x^{-(1+\alpha)}$, $\alpha\in(1/2,1)$, when $x$ goes to infinity, we prove that the following dynamic arises at the scaling limit: $\mathrm{d}\hat{P}_t = \sqrt{V_t}\mathrm{d}B^{(1)}_t$ with
\begin{align}\label{eq:full_quadratic_model_unstable}
V_t &= \frac{1}{\Gamma(\alpha)}\int_0^t (t-s)^{\alpha - 1}\lambda \big( \theta^0(s) + Z_s^2 - V_s\big)\mathrm{d}s +  \frac{1}{\Gamma(\alpha)}\int_0^t (t-s)^{\alpha - 1}\lambda \eta \sqrt{V_s}\mathrm{d}B^{(2)}_s\\
 Z_t &= \int_0^t k(t-s)\sqrt{V_s}\mathrm{d}B^{(1)}_s\nonumber,
\end{align}
with $\lambda$, $\eta$ some positive constants, $\theta^0$ a deterministic function and $(B^{(1)}, B^{(2)})$ two independent Brownian motions. As in the linear case, the near instability condition leads to appearance of a second Brownian motion driving a rough Heston type term. We see that the strong Zumbach effect is still reproduced thanks to the $Z^2_t$ term which is here convolved with a power-law kernel. Interestingly, we also show that when $k$ is regular, the $\mathrm{d}s$ term is proportional (up to a finite variation term) to $\int_0^th(t-s)Z_s\mathrm{d}Z_s$, where $h$ is a deterministic function with $h(0)<+\infty$. This can be interpreted as an essentially log-normal (non-rough) component, allowing us to view \eqref{eq:full_quadratic_model_unstable} as a super-Heston rough volatility model.

\section{Asymptotic behavior of purely quadratic Hawkes models}
\label{sec:purely_quadratic}
In this section we investigate the possible scaling limits of purely quadratic Hawkes based price processes. This corresponds to \eqref{eq:intensity_quadratic_hawkes} with $\phi = 0$. We devote a specific section to this case since it enables us to convey some of our main ideas in a simplified setting. More precisely, we consider $(N^T)_{T\geq 0}$ with intensity given by
\begin{equation}
\label{eq:purely_quadratic_intuition_a}
\lambda^T_t = \mu_T +  (Z^T_t)^2, \text{ with } Z^T_t = \int_{0}^t  k_T(t-s)\mathrm{d}P^T_s.
\end{equation}
For any $T$, the existence of the process $(N^T, P^T)$ can be obtained from \cite{jacod1975multivariate}. We are interested in the long time behavior of the price $P^T$ and of its intensity $\lambda^T$. Before stating the main result of this section we first discuss (in a non-rigorous manner) our scaling procedure.
\subsection{Scaling procedure}
The scaling procedure consists in finding appropriate factors $\omega_T$ so that the sequence $\omega_T \lambda^T_{tT}$ converges towards a non-degenerate limit. Assume $ \omega_T \lambda^T_{tT}$ converges towards some process $V_t$. Since $[P^T]_t = N^T_t$, we have
$$
\langle P^T \rangle_{Tt} =  \int_0^t T \lambda^T_{Ts}\mathrm{d}s.
$$
Thus we expect the martingale $P^{*T}_t = \sqrt{ \frac{\omega_T}{T}}P^T_{tT}$ to converge since its bracket does. Let $P$ be its limit. Since we wish to get $P$ continuous, we need $\omega_T/T$ to go to zero. From the convergence of $(\sqrt{\omega_T}Z^T_{tT})^2$, we expect that of
$$
\sqrt{\omega_T}Z^T_{tT} =  \int_0^t k_T\big( T(t-s)\big)\sqrt{T} \mathrm{d}P^{*T}_{s}
$$
too, which requires $k_T(T\cdot)\sqrt{T}$ to converge. This leads us to consider, as in \cite{blanc2017quadratic}, a sequence of kernels $k_T$ of the form
$$
k_T = k(\cdot / T) \sqrt{\gamma/T }
$$
for some $\gamma>0$ and $\omega_T=1$ (since we observe that $\omega_T$ plays eventually no role). Finally passing to the limit in \eqref{eq:purely_quadratic_intuition_a} we obtain the following candidate for our limiting process:
$$
V_t = \mu + Z_t^2,\text{ with }Z_t =\int_0^t k(t-s)\mathrm{d}P_s.
$$
\subsection{Assumptions and results in the purely quadratic case}

We now give our exact assumptions, the second of them being purely technical.
\begin{assumption}
\label{assumption:purely_quadratic}~\\

\noindent i) The sequence of kernels $(k_T)_{T \geq 0}$ is given by
$$
k_T = \sqrt{\frac{\gamma}{T}} k( \frac{\cdot}{T} ),
$$
with $\gamma \in (0, 1)$ and $k$ a non-negative measurable function such that $\|k\|_{2}=1$. Furthermore $\mu_T=\mu>0$.\\

\noindent ii)  The function $k$ belongs to  $L^{2+\varepsilon}$ for some $\varepsilon>0$ and for any $0\leq t < t'\leq 1$,  
$$
\int_0^t |k(t'-s) - k(t-s)|^2 \mathrm{d}s < C|t'-t|^{r},
$$
for some $r>0$ and $C>0$ and
$$
 \frac{1}{\eta}\int_0^1 |k(t)|^2 t^{-2\eta}\mathrm{d}t + \int_0^1\int_0^1 \frac{|k(t)-k(s)|^2}{|t-s|^{1+2\eta}} \mathrm{d}s\mathrm{d}t<+\infty
$$
for some $\eta \in (0,1)$.

\end{assumption}

\noindent Note that for $\alpha \in (1/2, 1)$ and $\lambda>0$, the Mittag-Leffler function $f^{\alpha, \lambda}$ satisfies Assumption \ref{assumption:purely_quadratic} $ii)$ for any $\varepsilon \in \big( 0, (2\alpha - 1)/(1-\alpha) \big)$, $\eta \in (0, \alpha - 1/2)$ and $r =2\alpha
 - 1 $.\\

\noindent Under Assumption \ref{assumption:purely_quadratic}, for any $T$, we have $\|k_T\|_2^2 = \gamma <1 $. So the stability condition is not violated at the limit. We now state the main result of this section. Consider the rescaled processes
$$
X^{T}_t = \frac{N^T_{tT}}{T}\text{ and }P^{*T}_t = \frac{1}{\sqrt{T}}P^T_{tT}.
$$
We have the following theorem.
\begin{theorem}
\label{th:pure_quadratic} Under Assumption \ref{assumption:purely_quadratic}, as $T$ goes to infinity, the sequence $(X^T, P^{*T})_{T\geq 0}$ converges in law for the Skorohod topology on $[0, 1]$ towards some processes $(X, P)$ satisfying the following properties:
\begin{itemize} 
\item $X$ is almost surely continuously differentiable.
\item There exists a Brownian motion $B$ such that $$P_t = \int_0^t \sqrt{V_s}\mathrm{d}B_s,$$
 where $V$ is the derivative of $X$ and the unique continuous solution of
\begin{equation}
\label{eq:purely_quadratic_limit}
V_t = \mu + Z_t^2,~~ Z_t =  \int_0^t \sqrt{\gamma} k(t-s)\sqrt{V_s}\mathrm{d}B_s,\text{ on }[0, 1].
\end{equation}
\item For any $\varepsilon>0$, if $k = f^{H + 1/2, \lambda}$ with $H\in (0, 1/2)$ and $\lambda>0$, $V$ has almost surely $H-\varepsilon$ H\"older regularity.
\end{itemize} 
\end{theorem}

\noindent Theorem \ref{th:pure_quadratic} will be generalized in Section \ref{sec:full_quadratic_stable} and its proof is given in Section \ref{proof:th_pure_quadratic}.

\subsection{Discussion of Theorem \ref{th:pure_quadratic}}\label{subsec:disc_pure_quad}

\noindent$\bullet~$Quadratic Hawkes models share many similarities with GARCH and QARCH models, see \cite{engle1982autoregressive, engle1986modelling, sentana1995quadratic}. However, from Theorem \ref{th:pure_quadratic}, we see that we do not need to be in the near instability regime $\|k_T\|^2_2 + \|\phi_T\|_1 \to 1$ in order to obtain a stochastic model at the scaling limit, while it is required in the GARCH setting, see \cite{nelson1990arch}.\\

\noindent$\bullet~$In the limiting model \eqref{eq:purely_quadratic_limit}, volatility and price are driven by the same Brownian motion $B$. This is in contrast to the GARCH case or to that of nearly unstable Hawkes processes where a new Brownian motion appears in the volatility dynamic, see \cite{el2018microstructural}. Compared to the GARCH situation, the difference essentially lies in the very constrained law of the returns here.\\

\noindent$\bullet~$The Zumbach effect is obviously present in the limiting model: the volatility is purely driven by the returns via the term $Z_t$.\\

\noindent$\bullet~$The use of Mittag-Leffler type kernels as in the last point of Theorem \ref{th:pure_quadratic} is very standard in the rough volatility literature, see for example \cite{Jaisson:2016aa}. It enables us to obtain at the limit a rough behavior for the sample paths of the volatility process.\\

\noindent$\bullet~$When $k(t) = \sqrt{ 2\nu} e^{-\nu t},$ Model \eqref{eq:purely_quadratic_limit} is that of \cite{blanc2017quadratic} with $\phi = 0$. Therefore Theorem \ref{th:pure_quadratic} extends the results of \cite{blanc2017quadratic} to any kernel $k$ with suitable integrability conditions. In the next section we provide an even more general extension that encompasses the case $\phi \neq 0$ and clearly shows the super-Heston nature of the dynamic \eqref{eq:purely_quadratic_limit}.

\section{General quadratic Hawkes models: the stable case}
\label{sec:full_quadratic_stable}
We now study the asymptotic behavior of a sequence of general quadratic Hawkes models for which the stability condition is not violated at the limit. We consider $(N^T)_{T\geq 0}$ with intensity given by \eqref{eq:intensity_quadratic_hawkes} (with parameters depending on $T$) where $\|\phi_T\|_1 + \|k_T\|_2^2$ is a fixed constant strictly smaller than one. As in the previous section, we first give intuitions about our scaling procedure.

\subsection{Suitable scaling in the general case}
Using a scaling factor $\omega_T$, the rescaled intensity becomes
\begin{align}
\nonumber \omega_T \lambda^T_{tT} = \mu_T \omega_T& + \int_0^t \phi_T\big( T(t-s)\big)T \omega_T\lambda^T_{Ts}\mathrm{d}s\\ 
& + \int_0^t \phi_T\big( T(t-s)\big)\sqrt{\omega_T T}\mathrm{d}(\sqrt{\frac{\omega_T }{T}}M^T_{Ts}) + (\sqrt{\omega_T}Z^T_{tT})^2,
\label{eq:general_stable_intuition_a}
\end{align}
where
$$
M^T_t = N^T_t - \int_0^{t} \lambda^T_s \mathrm{d}s.
$$
Assume that $(\omega_T \lambda^T_{tT})_{T\geq 0}$ converges and consider the processes $M^{*T}_t = M^T_{tT}\sqrt{\frac{\omega_T}{T}}$ and $P^{*T}_t = P^T_{tT}\sqrt{\frac{\omega_T}{T}}$. We have
$$
\langle P^{*T} \rangle_t = \langle M^{*T} \rangle_t = \int_0^t \omega_T\lambda^T_{Ts}\mathrm{d}s\text{ and }\langle P^{*T}, M^{*T} \rangle_t = 0.
$$
Thus we expect $P^{*T}$ and $M^{*T}$ to converge  towards two martingales $M$ and $P$ such that $\langle M, P\rangle=0$. As in the previous section, to obtain continuous martingales $M$ and $P$, we pick $\omega_T$ such that $\omega_T / T$ tends to zero.\\

\noindent One of our goals being to preserve Zumbach effect in the limit of \eqref{eq:general_stable_intuition_a}, we need a non-degenerate behavior for the feedback term $\sqrt{\omega_T}Z^{T}_{tT}$. We have
$$
\sqrt{\omega_T}Z^T_{tT} = \int_0^t k_T\big( T (t-s)\big)\sqrt{T}\mathrm{d}P^{*T}_t,
$$
which leads us again to the specification
$$
k_T = k(\cdot / T) \sqrt{ \gamma/T }
$$
for some positive $\gamma$. Now, if $\sqrt{\omega_T}Z^T_{tT}$ converges, according to \eqref{eq:general_stable_intuition_a} we should also obtain convergence of
\begin{equation*}
\label{eq:general_stable_intuition_b}
\mu_T\omega_T + \int_0^t \phi_T\big( T(t-s)\big)T \omega_T\lambda^T_{Ts}\mathrm{d}s + \int_0^t \phi_T\big( T(t-s)\big)\sqrt{\omega_T T}\mathrm{d}M^{*T}_s.
\end{equation*}
Thus, since both $\omega_T\lambda^T_{tT}$ and $M^{*T}$ are expected to converge, we set $\mu_T =\mu/\omega_T$ and must ensure the convergence of both $\phi_T(Tt)T$ and $\phi_T(Tt)\sqrt{\omega_T T}$. Because $\omega_T/T$ tends to zero, the first integral dominates the second one. Consequently we only need to take care of the first integral and again we can take $\omega_T=1$. A logical specification is therefore
$$
\phi_T = \phi(\cdot / T) (\beta/T)
$$
for some positive $\beta$. Passing to the limit in Equation \eqref{eq:general_stable_intuition_a} we expect the following limiting model:
$$
V_t = \mu + \int_0^t \beta \phi(t-s) V_s\mathrm{d}s + Z^2_t, \text{ with } Z_t = \int_0^t \sqrt{\gamma}k(t-s)\mathrm{d}P_s.
$$

\subsection{Assumptions and results in the presence of a linear component in the stable case}

We now give our exact assumptions which are quite similar to those in the previous section.

\begin{assumption}
\label{assumption:full_quadratic}~\\

\noindent i) The sequence of kernels is given by  
\begin{equation*}
k_{T}(t)=\sqrt{\frac{\gamma }{T}} k(\frac{t}{T}),~~ \phi_{T}(t)=  \frac{\beta}{T} \phi(\frac{t}{T}),
\end{equation*}
with $0<\gamma+\beta<1$ and $k$ and $\phi$ non-negative measurable such that $\|k\|_2^2 =\| \phi \|_1 = 1$. Furthermore $\mu_T=\mu>0$.\\

\noindent ii) Assumption \ref{assumption:purely_quadratic} ii) holds.

\end{assumption}

\noindent Assumption \ref{assumption:full_quadratic} implies that the stability condition is not violated at the limit. Nevertheless, from a rescaling perspective, the choice of kernels $\phi_T$ and $k_T$ does not seem really natural at first sight. It would be probably more satisfactory to consider kernel sequences of the form $a_T\phi$ and $a'_Tk$  (with $\phi$ and $k$ not depending on $T$) and then investigate the limit of 
$\omega_T \lambda^T_{tT}$ as in \cite{el2018microstructural,Jaisson:2015aa,Jaisson:2016aa}. This would imply here $\phi_T(Tt)T =a_T\phi(Tt)T$. According to Tauberian theorems, see for example \cite{bingham1989regular},  $a_T\phi(Tt)T $ can only converge in that case towards a power-law function of the form $t^{-\delta}$ for some positive $\delta$   and $\phi$ has to be such that $\phi(t)\sim_{+\infty}t^{-\delta}$ up to a slowly varying function. But recall that $\phi$ must be integrable and so we need $\delta \geq 1$. However such choice would lead to difficulties for defining properly the limit of the integral
$$
\int_0^t T a_T\phi\big(T(t-s)\big)\omega_T \lambda^T_{Ts}\mathrm{d}s.
$$
To be able to consider such types of natural but technically more intricate rescaling procedures, we will drop the stability assumption in Section \ref{sec:unstable_case} where we work in the nearly unstable case.\\

\noindent Let us define the rescaled process
$
X^T_t=N_{tT}/T.
$ We have the following theorem whose proof is given in Section~\ref{proof:full_quadratic_limit}.

\begin{theorem}
\label{th:full_quadratic_limit}Under Assumption \ref{assumption:full_quadratic}, the sequence $(X^T, P^{*T})_{T\geq 0}$ is $C$-tight for the Skorohod topology on $[0, 1]$ as $T$ goes to infinity, with the following properties for any limit point $(X,P)$:
\begin{itemize} 
\item $X$ is almost surely continuously differentiable.
\item There exists a Brownian motion $B$ such that $$P_t = \int_0^t \sqrt{V_s}\mathrm{d}B_s,$$
 where $V$ is the derivative of $X$ and the unique continuous solution of
\begin{equation}
\label{eq:general_quadratic_stable_limit}
V_t = \mu+ H_t + Z_t^2,\text{ with } H_t = \int_0^t \beta \phi(t-s) V_s\mathrm{d}s\end{equation} and $$
Z_t = \int_0^t \sqrt{\gamma}k(t-s)\sqrt{V_s}\mathrm{d}B_s,\text{ on }[0, 1].
$$
\item For any $\varepsilon>0$, if $k = f^{H + 1/2, \lambda}$ with $H\in (0, 1/2)$ and $\lambda>0$, $V$ has almost surely $H-\varepsilon$ H\"older regularity.
\end{itemize} 
\end{theorem}

\subsection{Discussion of Theorem \ref{th:full_quadratic_limit}}

\noindent$\bullet~$Compared to Theorem \ref{th:pure_quadratic}, only one new term appears in the volatility equation \eqref{eq:general_quadratic_stable_limit}. It comes from the self-exciting part in the Hawkes dynamic. Thus the elements in the discussion of the purely quadratic case in Section \ref{subsec:disc_pure_quad} remain valid here.\\

\noindent$\bullet~$Let us consider the case where $k$ is a continuously differentiable kernel with $0<k(0)<+\infty$. Using integration by parts and Fubini's theorem we can write
\begin{align*}
Z_t &=  \int_0^t \sqrt{\gamma}k(0)\sqrt{V_s}\mathrm{d}B_s + \int_0^t \int_u^{t} \sqrt{\gamma}k'(s-u) \mathrm{d}s \sqrt{V_u}  \mathrm{d}B_u\\
&=  \int_0^t \sqrt{\gamma}k(0)\sqrt{V_s}\mathrm{d}B_s + \int_0^t \int_0^{s} \sqrt{\gamma}k'(s-u)  \sqrt{V_u}  \mathrm{d}B_u \mathrm{d}s.
\end{align*}
Therefore $Z$ is a semi-martingale and up to a finite variation term we have
$$
Z_t^2 = \int_0^t 2\sqrt{\gamma}k(0)Z_s\sqrt{V_s}\mathrm{d}B_s.
$$
We see that the quadratic feedback term in the Hawkes dynamic induces a super-Heston type volatility because of the multiplicative term  $Z_s$ in front of the $\sqrt{V_s}$ in the equation above.\\

\noindent$\bullet~$Let us take the kernel $k$ as the Mittag-Leffler function $f^{\alpha, \lambda}$ with $\alpha\in (1/2, 1)$ and $\lambda>0$ and $\phi(t) = \kappa e^{-\kappa t}$ for some $\kappa>0$. Adapting Theorem 2.1 in \cite{el2018perfect} we get for any $h$ and $t_0$ positive
$$
Z_{t_0+h} = \xi_{t_0}(h) + \tilde{Z}_{h},\text{ with }\tilde{Z}_h = \int_{0}^h \sqrt{\gamma}f^{\alpha, \lambda}(h-s)\mathrm{d}P_{s+t_0}
$$
and
$$
\xi_{t_0}(h) = Z_{t_0} + \int_{0}^h f^{\alpha, \lambda}(h-s)\theta_{t_0}(s)\mathrm{d}s
$$
where
$$
\theta_{t_0}(h) = -Z_{t_0} + \frac{\alpha}{\lambda \Gamma(1-\alpha)}\int_{0}^{t_0}(t_0-s+h)^{-1-\alpha}(Z_s-Z_{t_0})\mathrm{d}s - \frac{(h+t_0)^{-\alpha}}{\lambda \Gamma(1-\alpha)}Z_{t_0}.
$$
Then we can write the forward volatility as
\begin{equation}
\label{eq:fwd_vol}
V_{t_0 + h} =  H_{t_0}e^{-\kappa h}  + \big( \xi_{t_0}(h) \big)^2 + 2\xi_{t_0}(h)\tilde{Z}_h + \mu + \tilde{H}_h + (\tilde{Z_h})^2
\end{equation}
with
$$
\tilde{H}_h = \int_0^h  \phi(h-s) V_{t_0+s}\mathrm{d}s.
$$
The function $\xi_{t_0}$ only depends on $(Z_t)_{0\leq t \leq t_0}$ and cannot be expressed as a function of $(V_t)_{0\leq t\leq t_0}$. So we get from \eqref{eq:fwd_vol} that conditional on the history of the market from time $0$ to $t_0$, the law of $(V_{t_0+h})_{h\geq 0}$ does depend on past returns and not only through past volatility. It means Models \eqref{eq:purely_quadratic_limit} and \eqref{eq:general_quadratic_stable_limit} can reproduce the strong Zumbach effect. So when $k$ is a Mittag-Leffler function, Model \eqref{eq:general_quadratic_stable_limit} is a super-Heston type rough volatility model with strong Zumbach effect.\\

\noindent In the case of exponential kernels $k(t) = \sqrt{2\nu}e^{-\nu t }$ and $\phi(t) = \kappa e^{-\kappa t}$ using similar computations we prove that 
$$
V_{t_0+h} = \mu + \tilde{Z}_h + \tilde{H}_t + e^{-2\nu h}Z_{t_0}^2 + 2Z_{t_0}e^{-\nu h} + e^{-\kappa h}H_{t_0}.\\
$$

\noindent$\bullet~$Finally remark that we do not prove uniqueness in law of the limit points $(X, P)$ in general. However we can show uniqueness in the special case $\phi = 0$. This particular case can be treated because $Z$ is the solution of a stochastic Volterra equation which admits a unique strong solution, see Section \ref{proof:th_pure_quadratic} for details and \cite{jaber2019affine} for more results about uniqueness of rough equations.

\section{Nearly unstable quadratic Hawkes models}
\label{sec:unstable_case}
We now focus on the case where the instability condition becomes almost violated at the limit. Let us consider a sequence of quadratic Hawkes processes $(N^T)_{T\geq 0}$ such that
$$
\|\phi_T\|_1 + \|k_T\|_2^2 \rightarrow 1.
$$
Contrary to the sections before, we wish to work here with a natural renormalization (at least for $\phi$, see comments below Assumption \ref{assumption:full_quadratic}) and therefore take $\phi_T$ of the form
$\phi_T =\beta_T\phi$ with $\beta_T\in (0, 1)$ and $\|\phi\|_1 = 1$. We also assume that $\phi$ is heavy-tailed ($\phi(x)\sim x^{-(1+\alpha)}$ with $\alpha\in (0,1)$ as $x$ tends to infinity) since this type of kernels leads to rough volatility in the case of linear Hawkes processes, see \cite{el2018microstructural,Jaisson:2016aa}.
Again, we start with insights about the suitable scaling procedure.

\subsection{An adapted scaling procedure in the nearly unstable case}
Let  $a_T = \|k_T\|_2^2 +\|\phi_T\|_1.$
We have
$$
\mathbb{E}[\lambda^T_t]= \mu_T +  \int_0^t (k^2_T + \phi_T )(t-s)\mathbb{E}[\lambda^T_s]\mathrm{d}s
$$
and therefore$$
\mathbb{E}[\lambda^T_t]\leq \frac{\mu_T}{1-a_T}.$$
So we naturally define the following renormalized processes:
$$
\lambda^{*T}_t = \frac{1 - a_T}{\mu_T}\lambda^T_{tT}, ~~\Lambda^{*T}_t = \frac{1}{T }\int_{0}^{tT} \lambda^{*T}_s\mathrm{d}s\text{ and }X^T_t  = \frac{1-a^T}{T \mu_T }N^T_{tT}.
$$
Let us assume that $\lambda^{*T}$ converges to some $V$. We can then expect that $\Lambda^{*T}$ and $X^{*T}$ converge to some $\Lambda$ and $X$. Consider the rescaled martingales
$$
M^{*T}_t = \sqrt{\frac{1-a_T}{T \mu_T}}M^T_{tT}\text{ and }P^{* T}_t = \sqrt{\frac{1-a_T}{T \mu_T}}P^T_{tT},
$$
where $ M^T_t = N^T_t - \int_0^t \lambda^T_s\mathrm{d}s$. Since $[M^T]_t =[P^T]_t = N^T_t,$ we have $[M^{*T}]_t = [P^{*T}]_t = X^T_t$. Moreover $\langle M^{*T},P^{*T}\rangle=0$ and so $M^{*T}$ and $P^{*T}$ are likely to converge towards some martingales $M$ and $P$ with same bracket $X$ and such that $\langle M,P\rangle = 0$.\\

\noindent Let $$\psi_T = \sum_{i\geq 1} \phi_T^{*i}.$$ Using Proposition 2.1 in \cite{Jaisson:2015aa}, we deduce from \eqref{eq:intensity_quadratic_hawkes} that 
\begin{equation*}
\lambda^T_t = \mu_T + (Z_t^T)^2 + \int_0^t \psi_T(t-s) ( \mu_T + (Z_s^T)^2 ) \mathrm{d}s + \int_0^t \psi_T(t-s) \mathrm{d}M^{ T}_s.
\end{equation*}
So we have
\begin{align}
\lambda^{*T}_t =& (1-a_T) + \frac{1-a_T}{\mu_T}(Z_{tT}^T)^2 + \int_0^t (1-a_T)T \psi_T\big(T(t-s)\big) ( 1 + \frac{1}{\mu_T}(Z_{sT}^T)^2 ) \mathrm{d}s \label{eq:intuition_limit_a}\\
& + \int_0^t T(1-a_T)\psi_T\big(T(t-s)\big) \frac{1}{\sqrt{T\mu_T(1-a_T)}} \mathrm{d}M^{*T}_{s}. \nonumber
\end{align}
The function $T \psi_T(T \cdot)$ has $L^1$ norm equal to 
$(1 - \beta_T)^{-1}$. Therefore $T(1-a_T)\psi_T(\cdot T)$ is non-vanishing only provided $1 - \beta_T$ is of order $1-a_T$. Consequently we set $\beta_T = 2a_T - 1$ (so that $\beta_T<a_T$).  Since $\|\phi_T\|_1=\beta_T \rightarrow  1$ then $\|k_T\|_2^2 \rightarrow  0$. However we will see that the sequence $k_T$ still plays a role in the limit.\\

\noindent In \eqref{eq:intuition_limit_a} the first integral is
\begin{equation*}
\label{eq:term_a}
\int_0^t T(1-a_T)\psi_T\big(T(t-s)\big)\mathrm{d}s.
\end{equation*}
It already appears in the case of a purely linear Hawkes process. We know from \cite{Jaisson:2015aa, jusselin2018no} that this term is crucial in the limiting behavior of the intensity and that a necessary condition to obtain a non-trivial scaling limit for it is that $T^{\alpha} (1-a_T)$ tends to a positive constant. Under this specification, we need to impose additionally that $T\mu_T(1-a_T)$  converges in order to obtain a non-degenerate asymptotic limit for the last integral in \eqref{eq:intuition_limit_a}.\\

\noindent We now study the terms containing the quadratic feedback:
$$
\frac{1-a_T}{\mu_T}(Z_{tT}^T)^2 \text{ and } \frac{1-a_T}{\mu_T} \int_0^t T \psi_T\big(T(t-s)\big)  (Z_{sT}^T)^2  \mathrm{d}s.
$$
Since $\|T \psi_T(\cdot T)\|_1 = (1-\beta_T)^{-1}$ which tends to infinity, the second term dominates the first one. To make the second term converge, we need a proper behavior of $Z^{*T}_{t}=Z^{T}_{tT}/\sqrt{\mu_T}$. We have
\begin{equation}
\label{eq:intuition_limit_b}
Z^{*T}_t = \sqrt{\frac{T}{1-a_T}}\int_0^t k_T\big(T(t-s) \big) \mathrm{d}P^{* T}_s.
\end{equation}
Thus we wish $ \sqrt{\frac{T}{1-a_T}} k_T(Tt)$ to converge and are naturally lead to assume that $k_T$ is of the form
$$
 k_T= k(\cdot/T)\sqrt{\frac{1-a_T}{T}}. 
$$

\subsection{Assumptions and results in the nearly unstable case}

\noindent We now summarize the conditions derived in the above discussion into the following assumption.
\begin{assumption}
\label{assumption:full_quadratic_unstable}~\\

\noindent i) The sequence of kernels $(\phi_T)_{T \geq 0}$ satisfies $ \phi_T = (2a_T-1) \phi$ with $(a_T)_{T \geq 0}$ a sequence in $(0,1)$ and $\phi$ a non-negative measurable function such that $\|\phi\|_1 = 1$. Furthermore for some $K>0$ and $\alpha \in (0, 1)$,
\begin{equation*}
\label{eq:tail_kernel}
\underset{x\rightarrow + \infty}{\lim}\alpha x^{\alpha}\int_{x}^{+\infty}\phi(s) \mathrm{d}s = K.
\end{equation*}
\noindent ii) The sequence of kernels $(k_T)_{T\geq 0}$ satisfies $k_T = k(\cdot /  T)  \sqrt{\frac{1-a_T}{T}}$ with $k$ a non-negative continuously differentiable function such that $\|k\|_2=1$ (in particular $k(0)<+\infty$).\\

\noindent iii) Let $\delta = K\frac{\Gamma(1-\alpha)}{\alpha}$. There are two positive constants $\lambda$ and $\mu^*$
such that
$$
 \underset{T\rightarrow +\infty}{\lim} (1 - a_T) T^{\alpha} =  \lambda \delta \text{ and }\underset{T\rightarrow +\infty}{\lim} T^{1 - \alpha} \mu_T = \mu^*\delta^{-1}.$$
\end{assumption}

\noindent The choice of $\delta$ in Point $iii)$ is just for convenience of notation in the results and proofs.\\

\noindent Recall that from Lemma 4.3 in \cite{Jaisson:2016aa}, under 
Assumption \ref{assumption:full_quadratic_unstable}, the function 
$$
F^T(t) = \int_0^tT(1-a_T)\psi_T(Ts)\mathrm{d}s
$$
converges towards $\frac{1}{2} F^{\alpha, \lambda}(t)$ where
$$
F^{\alpha, \lambda}(t) = \int_0^t  f^{\alpha, \lambda}(s) \mathrm{d}s.
$$
So $$T(1-a_T)\psi_T(Ts)\sim \frac{1}{2}f^{\alpha, \lambda}(s).$$
This provides us intuition for the form of the limit $(V,Z)$ of 
\eqref{eq:intuition_limit_a}-\eqref{eq:intuition_limit_b}:
$$
V_t =  \int_0^t \frac{1}{2} f^{\alpha, \lambda}(t-s) ( 1 + Z_s^2) \mathrm{d}s + \int_0^t \frac{1}{2} f^{\alpha, \lambda}(t-s) \frac{1}{\sqrt{\lambda\mu^*}} \mathrm{d}M_s,
$$
with $ Z_t = \int_0^t k(t-s)\mathrm{d}P_s $ and where $M$ and $P$ are martingales such that $\langle M, P\rangle= 0$ and $\langle M\rangle_t = \langle P\rangle_t = \int_0^t V_s\mathrm{d}s.$\\

\noindent We eventually state the main result of this section whose proof is given in Section \ref{proof:full_quadratic_unstable}.
\begin{theorem}
\label{th:full_quadratic_unstable}

Under Assumption \ref{assumption:full_quadratic_unstable}, the sequence $( X^T, M^{*T}, P^{*T} )_{T \geq 0}$ is $C$-tight for the Skorohod topology on $[0, 1]$ as $T$ goes to infinity, with the following properties for any limit point $(X,M,P)$:
\begin{itemize}
\item  We have $\langle M \rangle  = \langle P \rangle = X$ and  
\begin{equation}
\label{eq:general_quadratic_unstable_limit}
X_t = \int_0^t \frac{1}{2}F^{\alpha, \lambda}(t-s) \big( 1 + Z_s^2 \big)\mathrm{d}s + \int_0^t \frac{1}{2} f^{\alpha, \lambda}(t-s) \frac{1}{\sqrt{\lambda\mu^*}}M_s \mathrm{d}s 
\end{equation}
with
$$
Z_t = \int_0^t k(t-s) \mathrm{d}P_{s}.
$$
\item If $\alpha \in (1/2, 1),$ the process $X$ is almost surely continuously differentiable with derivative $V$ and up to an enlargement of the filtration there exists two Brownian motions $B^{(1)}$ and $B^{(2)}$ such that $V$ is solution of
\begin{equation*}
V_t = \int_{0}^t \frac{1}{2} f^{\alpha, \lambda}(t-s) \Big( \big( 1 + Z_s^2\big) \mathrm{d}s +  \frac{1}{\sqrt{\lambda\mu^*}}\sqrt{V_s} \mathrm{d}B^{(1)}_s \Big)
\end{equation*}
with
$$
Z_t = \int_0^t k(t-s)\sqrt{V_s}\mathrm{d}B^{(2)}_s.
$$
Moreover, for any $\varepsilon>0$, $V$ has almost surely $\alpha - \frac{1}{2}-\varepsilon$ H\"older regularity. 
\end{itemize}
\end{theorem}

\subsection{Discussion of Theorem   \ref{th:full_quadratic_unstable}}

\noindent$\bullet~$The form of the feedback is not the same in Model \eqref{eq:general_quadratic_unstable_limit} as in Model \eqref{eq:general_quadratic_stable_limit}. In \eqref{eq:general_quadratic_stable_limit} it is instantaneous through the $Z_t^2$ term while in  \eqref{eq:general_quadratic_unstable_limit} it is digested via a convolution with a fractional kernel.\\

\noindent$\bullet~$In Model \eqref{eq:general_quadratic_unstable_limit} price and volatility are driven by two different Brownian motions. This additional Brownian motion comes from the rescaling of the linear part of the intensity, as already observed for example in \cite{Jaisson:2015aa}.\\

\noindent$\bullet~$Rough volatility appears for very different reasons in Model \eqref{eq:general_quadratic_stable_limit} and Model \eqref{eq:general_quadratic_unstable_limit}. In \eqref{eq:general_quadratic_unstable_limit} the origin of rough volatility is the fat tail of the kernel $\phi$ while in \eqref{eq:general_quadratic_stable_limit} it arises from the behavior of the kernel $k$ in zero. Moreover it is clear from the proof of the last point of Theorem \ref{th:full_quadratic_unstable} that the regularity of $Z$ has no influence on that of $V$.\\

\noindent$\bullet~$As computed in the previous section we can write
\begin{equation*}
Z_t =  \int_0^t k(0)\sqrt{V_s}\mathrm{d}B^{(2)}_s + \int_0^t \int_0^{s} k'(s-u)  \sqrt{V_u}  \mathrm{d}B^{(2)}_u \mathrm{d}s.
\end{equation*}
Therefore $Z$ is a semi-martingale and up to a finite variation term we have 
$$
 \mathrm{d}Z_t^2 =  2k(0) Z_t \sqrt{V_s}\mathrm{d}B^{(2)}_s.
$$
Furthermore using integration by parts we get
$$
\int_0^t f^{\alpha, \lambda}(t-s)Z_s^2\mathrm{d}s = \int_0^t F^{\alpha, \lambda}(t-s)\mathrm{d}Z_s^2.
$$
So up to a finite variation term, we have in Model \eqref{eq:general_quadratic_unstable_limit}
$$
V_t = \int_0^t f^{\alpha, \lambda}(t-s) \frac{1}{\sqrt{\lambda\mu^*}}\sqrt{V_s} \mathrm{d}B^{(1)}_s  + \int_{0}^t  F^{\alpha, \lambda}(t-s) k(0)Z_s\sqrt{V_s}\mathrm{d}B^{(2)}_s.
$$
Thus as in Model \eqref{eq:purely_quadratic_limit} and \eqref{eq:general_quadratic_stable_limit}, the quadratic feedback term in the volatility dynamic induces that Model \eqref{eq:general_quadratic_unstable_limit} is a super-Heston type rough volatility model. Note however that in that case, the super-Heston and rough components are not the same.\\

\noindent$\bullet~$Using Lemma A.2 in \cite{el2019characteristic}, when $\alpha\in (1/2, 1)$, we get that Equation \eqref{eq:general_quadratic_unstable_limit} is equivalent to
$$
V_t = V_0 + \frac{1}{\Gamma(\alpha)}\int_0^t (t-s)^{\alpha - 1}\lambda(Z_s^2 + \theta_0(s) - V_s)\mathrm{d}s + \frac{1}{\Gamma(\alpha)}\int_0^t (t-s)^{\alpha - 1}\sqrt{V_s}\frac{1}{\sqrt{\lambda\mu^*}}\mathrm{d}B^{(1)}_s
$$
with $\theta_0$ a deterministic function. In the case $k(t) = \sqrt{2\nu}e^{-\nu t}$ with $\nu>0$, direct adaptation of Theorem 2.1 in \cite{el2018perfect} gives that
\begin{align*}
V_{t_0+h} = V_{t_0}&  + \frac{1}{\Gamma(\alpha)}\int_0^h (h-s)^{\alpha - 1}\lambda(\tilde{Z}_s^2 + 2\tilde{Z}_sZ_{t_0}e^{-\nu s}  +  \theta_{t_0}(s) - V_{t_0+s})\mathrm{d}s \\
& + \frac{1}{\Gamma(\alpha)}\int_0^h (h-s)^{\alpha - 1}\sqrt{V_{t_0 +s }}\frac{1}{\sqrt{\lambda\mu^*}}\mathrm{d}B^{(1)}_s,
\end{align*}
where $\theta_{t_0}(h)$ is equal to
$$
\theta_0(t_0 + h) + \frac{\alpha}{\Gamma(1-\lambda)}\int_0^{t_0}(t_0 -v +h)^{-1-\alpha}(V_v - V_{t_0})\mathrm{d}v + \frac{(t + t_0)^{-\alpha}}{\lambda \Gamma(1-\alpha)}(V_{0} - V_{t_0}) + Z_{t_0}^2 e^{-2\nu h}
$$
and 
$$
\tilde{Z}_h = \int_0^h k(h-s)\mathrm{d}P_{t_0 +s}.
$$
The term $Z_{t_0}$ cannot be written as a function of $(V_t)_{0\leq t \leq t_0}$. So Model \eqref{eq:general_quadratic_unstable_limit} reproduces the strong Zumbach effect. Finally Model \eqref{eq:general_quadratic_unstable_limit} is a super-Heston type rough volatility model with strong Zumbach effect.\\

\noindent$\bullet~$ Despite the many recent works about stochastic Volterra equations, see for example \cite{abi2017affine, jaber2019sve}, the existence of a strong solution to equations such as \eqref{eq:general_quadratic_unstable_limit} remains an open question.

\section{Proofs}
\label{Proofs}

We gather all the proofs in this section. We first show Theorem \ref{th:pure_quadratic} assuming that Theorem \ref{th:full_quadratic_limit} holds. Then we give the proof of Theorem \ref{th:full_quadratic_limit} and finally that of Theorem \ref{th:full_quadratic_unstable}.

\subsection{Proof of Theorem \ref{th:pure_quadratic}}
\label{proof:th_pure_quadratic}

Using the results of Theorem \ref{th:full_quadratic_limit}, we only need to prove that when $\phi = 0$, the limiting process $(X, P)$ in Theorem \ref{th:full_quadratic_limit} $ii)$ is unique in law.\\

\noindent Consider $(X, P)$ a limit point of $(X^T, P^{*T})_{T\geq 0}$. Then $V$ the derivative of $X$ satisfies
$$
V_t = \mu + Z_t^2
$$
and there exists a Brownian motion $B$ such that $P_t = \int_0^t\sqrt{V_s}\mathrm{d}B_s$. Thus we can write
\begin{equation}
\label{eq:proof_pure_a}
Z_t = \int_0^t k(t-s)\sqrt{\mu + Z_s^2}~\mathrm{d}B_s.
\end{equation}
From Assumption \ref{assumption:purely_quadratic} $ii)$ together with Theorems 3.1 and 3.3 in \cite{zhang2010stochastic}, there is a unique process $Z$ satisfying \eqref{eq:proof_pure_a} and it is continuous. Since $P$ is a continuous martingale satisfying $[ P ] = X$, $X = \int_0^t V_s\mathrm{d}s$ and $V_t = \mu + Z_t^2$, the limiting process $(X, P)$ is fully determined by the only solution of \eqref{eq:proof_pure_a}. So we get convergence of $(X^T, P^{*T})$ for the Skorohod topology.

\subsection{Proof of Theorem~ \ref{th:full_quadratic_limit}}
\label{proof:full_quadratic_limit}

We proceed in three steps. First we prove that the sequence $(X^T, P^{*T})_{T\geq 0}$ is $C$-tight for the Skorohod topology. Then we show the results about the dynamics of the limit points. Finally we establish the regularity properties of the limit points.

\subsubsection{Tightness of the sequence $(X^T, P^{*T})_{T\geq 0}$}

We consider the processes
$$
\Lambda^{*T}_t = \int_{0}^t \lambda^{*T}_s\mathrm{d}s \text{ and }Z^{*T}_t = \int_0^t k(t-s)\mathrm{d}P^{*T}_s
$$
defined for $t\in [0, 1]$. Remark that $\Lambda^{*T}$ is the predictable compensator of $X^T$. We have the following equality:
\begin{equation*}
\label{eq:intensity_inequality}
\mathbb{E}[\lambda^T_{t}] = \mu_T + \int_{0}^t  \big(k_T^{2}+\phi_T \big)(t-s)\mathbb{E}[\lambda^T_{s}]\mathrm{d}s. 
\end{equation*}
Thus 
$$
\mathbb{E}[\lambda^T_t]\leq \frac{ \mu }{1 - \|\phi_T\|_1- \|k_T\|_2^2 } 
$$
and consequently
$$
\mathbb{E}[X^T_1] = \mathbb{E}[\Lambda^{*T}_1] \leq \frac{ \mu }{1 - \gamma-\beta}.
$$
Since the processes $X^T$ and $\Lambda^{*T}$ are increasing for any $T$, using the last inequality, we deduce from Theorem VI-3.21 together with Proposition VI-3.35 in \cite{jacod2013limit} that $(X^T)_{T\geq 0}$ and $(\Lambda^{*T})_{T\geq 0}$ are tight. Moreover since $|\Delta X^T| +|\Delta \Lambda^{*T}|\leq 1/T$ almost surely on $[0, 1]$, according to Proposition VI-3.26 in \cite{jacod2013limit}, $(X^T)_{T\geq 0}$ and $(\Lambda^{*T})_{T\geq 0}$ are $C$-tight. The tightness of $(M^{*T})_{T\geq 0}$ and $(P^{*T})_{T\geq 0}$ follows from Theorem VI-4.13 in \cite{jacod2013limit} using that $\langle M^{*T}\rangle_t = \langle P^{* T}\rangle_t = \Lambda^{*T}_t$. We then get $C$-tightness because $|\Delta M^{*T}|+|\Delta P^{*T}|\leq 2/T$. Finally $(X^T, \Lambda^{*T}, M^{* T}, P^{* T})_{T\geq 0}$ is $C$-tight for the Skorohod topology on $[0, 1]$.\\

\noindent We also show that the sequence $(Z^{*T})_{T\geq 0}$ is tight for the $L^2([0, 1])$ topology. For this, inspired by \cite{jaber2019sve}, we consider the Sobolev-Slobodeckij norm defined for any measurable function $f$ by
$$
\|f\|_{W^{\eta, 2}([0, 1])} = \Big( \int_0^1f(s)^2\mathrm{d}s + \int_0^1\int_0^1 \frac{|f(t) - f(s)|^2}{|t-s|^{1+2\eta}}\mathrm{d}s\mathrm{d}t \Big)^{1/2}.
$$
We recall that the closed balls of $\|\cdot\|_{W^{\eta, 2}([0, 1])}$ are relatively compact in $L^2([0, 1])$, see \cite{flandoli1995martingale}. Therefore it is enough to show that $\big( \mathbb{E}[\|Z^{*T}\|^2_{W^{\eta, 2}([0, 1])}] \big)_{T\geq 0 }$ is uniformly bounded to conclude the tightness of $(Z^{*T})_{T\geq 0}$ in $L^2([0, 1])$.\\

\noindent For any $t\in [0, 1],$ we have using Ito's formula
$$
\mathbb{E}[(Z^{*T}_t)^2] = \int_0^t k^2(t-s)\mathbb{E}[\lambda^T_{Ts}]\mathrm{d}s \leq \frac{\mu}{1 - \gamma-\beta} \|k\|_2^2
$$
and for $0\leq s \leq t\leq 1$
$$
Z^{*T}_t - Z^{*T}_s= \int_{s}^t k(t-u)\mathrm{d}P^{*T}_u + \int_0^s \big( k(t-u)-k(s-u) \big) \mathrm{d}P^{*T}_u.
$$
Then we get
$$
\mathbb{E}[(Z^{*T}_t - Z^{*T}_s)^2] = \int_{s}^t k^2(t-u)\mathbb{E}[\lambda^T_{Tu}]\mathrm{d}u + \int_0^s \big( k(t-u)-k(s-u)\big)^2 \mathbb{E}[\lambda^T_{Tu}]\mathrm{d}u.
$$
Using that $\mathbb{E}[\lambda^T_u] \leq \frac{\mu}{1 - \gamma-\beta}$ we obtain
$$
\mathbb{E}[(Z^{*T}_t - Z^{*T}_s)^2] \leq \frac{\mu}{1 - \gamma-\beta}\Big( \int_{s}^t k^2(t-u)\mathrm{d}u + \int_0^s \big(k(t-u)-k(s-u)\big)^2 \mathrm{d}u\Big).
$$
According to \cite{jaber2019sve} we have
$$
\int_0^1 \int_0^1 \int_{s\wedge t }^{s \vee t} \frac{k(s\vee t - u)^2}{|t-s|^{1+2\eta}}\mathrm{d}u\mathrm{d}s\mathrm{d}t\leq \frac{1}{\eta}\int_0^1 |k(t)|^2 t^{-2\eta}\mathrm{d}t
$$
and
$$
\int_0^1 \int_0^1 \int_{s\wedge t }^{s \vee t} \frac{|k(t - u) - k(s - u)|^2}{|t-s|^{1+2\eta}}\mathrm{d}u\mathrm{d}s\mathrm{d}t\leq \int_0^1\int_0^1 \frac{|k(t)-k(s)|^2}{|t-s|^{1+2\eta}} \mathrm{d}s\mathrm{d}t,
$$
which is bounded from Assumption \ref{assumption:full_quadratic} $ii)$. Finally using Fubini's theorem we deduce that $(\mathbb{E}[\|Z^{*T}\|^2_{W^{\eta, 2}([0, 1])}])_{T\geq 0}$ is bounded. So $(Z^{*T})_{T\geq 0}$ is tight in $L^2([0, 1])$.\\

\noindent Before going to the next step we prove the following lemma.
\begin{lemma}
\label{lemma:ucp}  The sequence of martingales $X^{T}-\Lambda^{*T}$ converges to $0$ uniformly in probability on $[0,1]$.
 \end{lemma} 
 \begin{proof}
Since $N^T_t - \int_0^t \lambda^T_s\mathrm{d}s$ is a true martingale, from Doob's inequality we get 
$$
\mathbb{E}[ \sup_{t \in [0,1]}  (X_t^{T}-\Lambda_t^{T*})^2 ]\leq \frac{1}{T^2} \mathbb{E}[N^{T}_T].
$$
Using that $\mathbb{E}[N^T_T] = T \mathbb{E}[ \Lambda^{*T}_1]$ we deduce
$$
\mathbb{E}[ \sup_{t \in [0,1]}  (X_t^{T}-\Lambda_t^{*T})^2 ] \leq \frac{1}{T} \frac{\mu }{1 - \gamma-\beta},
$$
which concludes the proof.
\end{proof}

\subsubsection{Dynamic of the limit points}
\label{subsubsection:proof_a_ii}

We now consider $(X, X, M, P, Z)$ a limit point of $(X^T, \Lambda^{*T}, M^{*T}, P^{*T}, Z^{*T})_{T\geq 0}$. Using Skorohod representation theorem and the fact that $(X, X, M, P)$ is continuous, we may consider that almost surely $(X^T, \Lambda^{*T}, M^{*T}, P^{*T})_{T\geq 0}$ converges uniformly on $[0, 1]$ towards $(X, X, M, P)$ and $(Z^{*T})_{T\geq 0}$ converges in $L^2([0, 1])$ towards $Z$:
\begin{align}
\label{eq:proof_unif_conv}
 \underset{t\in [0, 1]}{\sup}|X^T_t - X_t|\underset{T\rightarrow +\infty}{\rightarrow}  0&,~~\underset{t\in [0, 1]}{\sup}|M^{*T}_t -M_t|  \underset{T\rightarrow +\infty}{\rightarrow}  0,\\
\nonumber \underset{t\in [0, 1]}{\sup}|P^{* T}_t -P_t| \underset{T\rightarrow +\infty}{\rightarrow}  0&\text{ and } \int_0^1 (Z^{*T}_s - Z_s)^2\mathrm{d}s\underset{T\rightarrow +\infty}{\rightarrow}  0.
\end{align}
From Corollary IX-1.19 in \cite{jacod2013limit} we have that $M$ and $P$ are local martingales. Moreover $[ M^{*T} ] = [ P^{* T}] = X^T$ so Corollary VI-6.29 in \cite{jacod2013limit} gives that $[ M ] = [ P ]  = X $. Since $M$ and $P$ are continuous we have 
$$
\langle M\rangle = [ M ] = \langle P \rangle = [P]  = X.
$$
 We also note that $\mathbb{E}[X^T_1]$ is uniformly bounded in $T$. So from Fatou's lemma $X_1$ is integrable and $M$ and $P$ are true martingales. Moreover up to a subsequence $(Z^{*T})_{T\geq 0}$ converges almost surely towards $Z$. We deduce that $Z$ is adapted. Moreover since $\mathbb{E}[\underset{t\in [0, 1]}{\sup}(Z^{*T}_t)^2]$ is bounded we get $$\underset{t\in [0, 1]}{\sup}\mathbb{E}[Z_t^2]<+\infty.$$
\noindent We show that $\Lambda^{*T}$ converges almost surely uniformly on $[0, 1]$ towards
\begin{equation}
\label{eq:proof_aux_g}
\int_0^t \big(\mu + Z_s^2\big)\mathrm{d}s +\int_0^t F_1(t-s)\mathrm{d}X_{s},
\end{equation}
where $F_1(t) = \int_0^t \beta \phi(s)\mathrm{d}s$. We have
\begin{equation*}
\Lambda^{*T}_t = \int_0^t \big(\mu + (Z^{*T}_s)^2\big)\mathrm{d}s + \int_0^t \int_0^t \beta \phi(s-u)\mathrm{d}X^T_u\mathrm{d}s.
\end{equation*}
The almost sure convergence of $(Z^{*T})_{T\geq 0}$ in $L^2$ towards $Z$ implies that almost surely, uniformly in $t\in [0, 1]$,
$$
\int_0^t (Z^{*T}_s)^2 \mathrm{d}s \underset{T\rightarrow +\infty}{\rightarrow}  \int_0^t Z_s^2 \mathrm{d}s.
$$
Moreover using Ito's formula together with Fubini's theorem we get
$$
\int_0^t \int_0^s \beta \phi(s-u)\mathrm{d}X^{T}_u\mathrm{d}s = \int_0^t \beta \phi(t-s)X^T_s\mathrm{d}s.
$$
From Equation \eqref{eq:proof_unif_conv}, we deduce that this quantity converges almost surely uniformly towards 
$$
\int_0^t \beta \phi(t-s)X_s\mathrm{d}s.
$$
Again Ito's formula together with Fubini's theorem give
$$
\int_0^t \beta \phi(t-s)X_s\mathrm{d}s  =  \int_0^t F_1(t-s)\mathrm{d}X_s.
$$
So we obtain the almost sure uniform convergence of $(\Lambda^{*T})_{T\geq 0}$ towards \eqref{eq:proof_aux_g}. Consequently, using Lemma \ref{lemma:ucp}, we deduce
$$
X_t = \int_0^t \mu + Z_s^2 \mathrm{d}s + \int_0^t F_1(t-s)\mathrm{d}X_s
$$
and eventually
$$
X_t =  \int_0^t \mu + Z^2_s + \int_0^s \beta \phi(s-u)\mathrm{d}X_u \mathrm{d}s.
$$
Thus $X$ is absolutely continuous with respect to the Lebesgue measure with derivative $V$ given by
\begin{equation*}
\label{eq_aux_h}
V_t =  \mu + Z^2_t + \int_0^t \beta \phi(t-s) V_s \mathrm{d}s.
\end{equation*}
Letting $\psi = \sum_{i\geq 1} (\beta \phi)^{*i}$ we have
\begin{equation}
\label{eq_aux_g}
V_t = \mu + Z_t^2 + \int_{0}^t \psi(t-s)Z_s^2\mathrm{d}s.
\end{equation}
The boundedness of $(\mathbb{E}[Z_t^2])_{t\in [0, 1]}$ gives that $(V_t)_{t\in[0, 1]}$ is uniformly bounded in $L^1$.\\

\noindent We now prove that 
$$
Z_t = \int_{0}^tk(t-s)\mathrm{d}P_s.
$$
Using Cauchy-Schwarz inequality, the convergence of $(Z^{*T})_{T\geq 0}$ implies that almost surely, uniformly in $t\in [0, 1],$
$$
\int_0^t Z^{*T}_s\mathrm{d}s \underset{T\rightarrow +\infty}{\rightarrow} \int_0^t Z_s\mathrm{d}s.
$$
From Ito's formula we get 
$$
\int_0^t Z^{*T}_s\mathrm{d}s = \int_0^t k(t-s)P^{*T}_s\mathrm{d}s
$$
and using Equation \eqref{eq:proof_unif_conv} we deduce that it converges almost surely uniformly towards 
$$
\int_0^t k(t-s)P_s\mathrm{d}s.
$$
Since $F_2(t) = \int_0^t \gamma k^2(s)\mathrm{d}s<1$ we have 
$$
\int_0^t \int_0^s k(s-u)^2\mathrm{d}X_u\mathrm{d}s = \int_0^t F_2(t-s)\mathrm{d}X_s < X_t < +\infty.
$$
So we can use the stochastic Fubini theorem and show that
$$
\int_0^t \int_0^s k(s-u)\mathrm{d}P_u\mathrm{d}s = \int_0^t k(t-s)P_s\mathrm{d}s.
$$
Thus almost surely, for any $t\in [0, 1],$
\begin{equation*}
\int_{0}^t Z_s\mathrm{d}s = \int_0^t \int_0^s k(s-u)\mathrm{d}P_u\mathrm{d}s
\end{equation*}
and
\begin{equation*}
Z_t = Z_0 + \int_0^t \sqrt{\gamma} k(t-s)\mathrm{d}P_s.
\end{equation*}
Moreover, from Theorem V-3.9 in \cite{revuz2013continuous}, there exists a Brownian motion $B$ such that
$$
P_t = \int_0^t \sqrt{V_s}\mathrm{d}B_s
$$
and finally we get
\begin{equation*}
\label{eq:volterra}
Z_t = \int_0^t \sqrt{\gamma} k(t-s)\sqrt{V_s}\mathrm{d}B_s.
\end{equation*}

\noindent We recall that $(\mathbb{E}[V_t])_{t\in [0, 1]}$ is bounded in $L^2$. So using  Assumption \ref{assumption:full_quadratic} $ii)$ together with Theorems 3.1 and 3.3 in \cite{zhang2010stochastic} we obtain that the process $Z$ is continuous. Therefore using \eqref{eq_aux_g} $V$ is also continuous. This concludes this part of the proof.

\subsubsection{Regularity property}

We now consider that $k$ is given by $f^{H+1/2, \lambda}$ for $H\in (0, 1/2)$ and $\lambda>0$. We can write 
$$
\int_0^t Z_s\mathrm{d}s = \int_0^t \sqrt{\gamma}f^{H+1/2, \lambda}(t-s)P_s\mathrm{d}s.
$$
Since $P_t = \int_{0}^t \sqrt{V_s}\mathrm{d}B_s$, $P$ has the same regularity as a Brownian motion. Thus we can use the same arguments as in Section 4.4 in \cite{Jaisson:2016aa} to deduce that $Z$, and therefore $V$, are $H-\varepsilon$ H\"older for any $\varepsilon>0$.

\subsection{Proof of Theorem~\ref{th:full_quadratic_unstable}}
\label{proof:full_quadratic_unstable}

We proceed again in three steps. First we show that the sequence $(X^T, M^{*T}, P^{*T})_{T\geq 0}$ is $C$-tight for the Skorohod topology. Then we prove the results about the dynamics of the limit points and finally those on the regularity of the limit points.

\subsubsection{Tightness of $(X^T, M^{*T}, P^{*T})_{T\geq 0}$}
Recall the definition of the renormalized processes
$$
\lambda^{*T}_t = \frac{1 - a_T}{\mu_T}\lambda^T_{tT}, ~~\Lambda^{*T}_t = \frac{1-a^T}{T \mu_T}\int_{0}^{tT} \lambda^T_s\mathrm{d}s,~~X^T_t  = \frac{1-a^T}{T \mu_T }N^T_{tT},~~Z^{*T}_t = Z^{T}_{tT}/\sqrt{\mu_T}
$$
$$
M^{*T}_t = \sqrt{ \frac{1-a_T}{T \mu_T}} M^T_{tT} \text{ and } P^{*T}_t = \sqrt{ \frac{1-a_T}{T \mu_T}} P^{T}_{tT}.
$$

\noindent We have 
$$
\mathbb{E}[\lambda^T_t]\leq \mu_T + \int_0^t \big( k_T^2(t-s) + \phi_T(t-s) \big)\mathbb{E}[\lambda^T_s] \mathrm{d}s.
$$
Thus
$$
\mathbb{E}[\lambda^T_t]\leq \frac{ \mu_T }{ 1 - \|\phi_T\|_1 - \|k_T\|_2^2}
$$
and consequently
$$
\mathbb{E}[\lambda^{*T}_t]\leq  \frac{ 1-a_T }{ 1 - \beta_T - \|k_T\|_2^2 } = 1.
$$
So
$$
\mathbb{E}[X^{T}_1]=\mathbb{E}[\Lambda^{*T}_1] \leq 1,
$$
which gives the tightness of the sequences $\big((X^{T}_t)_{t\in [0, 1]}\big)_{T\geq 0}$ and $\big( (\Lambda^{*T}_t)_{t\in [0, 1]}\big)_{T\geq 0}$, both of them being increasing. Actually we get $C$-tightness since  $|\Delta X^{T}| + |\Delta \Lambda^{*T}|\leq \frac{1-a_T}{T \mu_T }$ that goes to zero as $T$ goes to infinity. Remark that Lemma \ref{lemma:ucp} still holds under Assumption \ref{assumption:full_quadratic_unstable}.\\

\noindent The tightness of $(M^{*T})_{T\geq 0}$ and $(P^{*T})_{T\geq 0}$ follows from Theorem VI-4.13 in \cite{jacod2013limit} because $\langle M^{*T}\rangle_t = \langle P^{* T}\rangle_t = \Lambda^{*T}_t$ and $\langle M^{*T}, P^{*T}\rangle = 0$. We then obtain $C$-tightness since $|\Delta M^{*T}|+|\Delta P^{*T}|\leq 2/T$. Finally $(X^T, \Lambda^{*T}, M^{* T}, P^{* T})_{T\geq 0}$ is $C$-tight for the Skorohod topology on $[0, 1]$.

\subsubsection{Dynamics of the limit points}

We now take $(X, X, M, P)$ a limit point of $(X^T, \Lambda^{*T}, M^{*T}, P^{*T})_{T\geq 0}$. Since $(X, X, M, P)$ is continuous, according to the Skorohod representation theorem, we can consider that $(X^T, \Lambda^{*T}, M^{*T}, P^{*T})_{T\geq 0}$ converges almost surely uniformly towards $(X, X, M, P)$:
\begin{equation*}
\underset{t\in [0, 1]}{\sup}|X^T_t - X_t|\underset{T\rightarrow +\infty}{\rightarrow}  0,~~\underset{t\in [0, 1]}{\sup}|\Lambda^{*T}_t - X_t|\underset{T\rightarrow +\infty}{\rightarrow}  0,
\end{equation*}
and
\begin{equation}
\label{eq:proof_aux_e}
\underset{t\in [0, 1]}{\sup}|M^{* T}_t - M_t|\underset{T\rightarrow +\infty}{\rightarrow}  0,~~\underset{t\in [0, 1]}{\sup}|P^{* T}_t - P_t|\underset{T\rightarrow +\infty}{\rightarrow}  0.
\end{equation}
From Corollary IX-1.19 in \cite{jacod2013limit}, we have that $M$ and $P$ are local martingales. Moreover since $[ M^{*T} ] = [ P^{* T}] = X^T$, we have $[ M ] = [ P ]  = X $ and $\langle M, P\rangle = 0$ using Corollary VI-6.29 in \cite{jacod2013limit}. Because $M$ and $P$ are continuous, we deduce
$$
\langle M\rangle = [ M ] = \langle P\rangle = [ P ]  = X.
$$
Because $\mathbb{E}[X^T_1]$ is uniformly bounded in $T$, we get that $X_1$ is in $L^1$ and so $M$ and $P$ are true martingales. In addition, the Dambis-Dubin-Schwarz theorem gives the existence of two independent Brownian motions $B^{(1)}$ and $B^{(2)}$ such that
$$
M_t = B^{(1)}_{X_t} \text{ and }P_t= B^{(2)}_{X_t}.
$$

\noindent Recall that for $F^T(t) = \int_{0}^t T(1-a_T)\psi_T(Ts)\mathrm{d}s$ we have
\begin{eqnarray*}
 \Lambda^{*T}_t&=& t(1-a_T)+ \int_{0}^{t}F^T(t-s)ds+\int_{0}^{t}\frac{F^T(t-s)}{\sqrt{T(1-a_T)\mu_{T}}}dM^{*T}_s \\
 && +\int_{0}^{t} F^T(t-s)(Z^{*T}_s)^{2}ds + \int_0^t (1-a_T)(Z_s^{*T})^2\mathrm{d}s.
\end{eqnarray*} 
According to Lemma 4.3 in \cite{Jaisson:2016aa}, we have the uniform convergence
$$
\int_{0}^{t}F^T(t-s)ds \underset{T \rightarrow +\infty}{\rightarrow} \int_{0}^{t}\frac{1}{2}F^{\alpha, \lambda}(t-s)ds.
$$
Using integration by parts we obtain
$$
Z^{*T}_t =  k(0) P^{*T}_t + \int_0^t k'(t-s) P^{* T}_s\mathrm{d}s.
$$
Assumption \ref{assumption:full_quadratic_unstable} $i)$ implies that $k'$ is bounded on $[0, 1]$. As a consequence of \eqref{eq:proof_aux_e} we have that almost surely, $Z^{*T}$ converges uniformly on $[0, 1]$ towards
$$
k(0) P_t + \int_0^t k'(t-s) \mathrm{d}P_s\mathrm{d}s = \int_0^t k(t-s)\mathrm{d}P_s
$$
which is continuous. This convergence together with Lemma 4.3 in \cite{Jaisson:2016aa} implies that almost surely, uniformly in $t\in [0, 1]$,
$$
\int_{0}^{t} F^T(t-s)(Z^{*T}_s)^{2}ds \underset{T\rightarrow +\infty}{\rightarrow}  \int_{0}^{t} F^{\alpha, \lambda}(t-s)Z_s^{2}ds
$$
and
$$
(1-a_T)\int_0^t (Z^{*T}_s)^2 \mathrm{d}s \underset{T\rightarrow +\infty}{\rightarrow}  0.
$$
We now prove that $\int_{0}^{t}\frac{F^T(t-s)}{\sqrt{T(1-a_T)\mu_{T}}}dM^{*T}_s$ converges uniformly in probability towards
$$
 \int_{0}^{t}\frac{f^{\alpha, \lambda}(t-s)}{2\sqrt{\lambda \mu^*}}M_s \mathrm{d}s.
$$
Using integration by parts we have
$$
\int_{0}^{t}\frac{F^T(t-s)}{\sqrt{T(1-a_T)\mu_{T}}}\mathrm{d}M^{*T}_s = \int_{0}^{t}\frac{f^T(t-s)}{\sqrt{T(1-a_T)\mu_{T}}}M^{*T}_s\mathrm{d}s.
$$
Remark that 
$$
\int_{0}^{t} f^T(t-s) M^{*T}_s \mathrm{d}s - \int_{0}^{t}f^{\alpha, \lambda}(t-s)\frac{1}{2} M_s\mathrm{d}s
$$
can be written
\begin{equation}
\label{eq:proof_aux_f}
\int_0^t \frac{1}{2}f^{\alpha, \lambda}(t-s) (M^{*T}_s -M_s)\mathrm{d}s + \int_0^t \big( f^{T}(t-s) - \frac{1}{2}f^{\alpha, \lambda}(t-s)\big)M^{*T}_s\mathrm{d}s.
\end{equation}
The first term in \eqref{eq:proof_aux_f} goes almost surely uniformly to zero using \eqref{eq:proof_aux_e} and the fact that $f^{\alpha, \lambda}\in L^1$. Applying integration by parts again we obtain
$$
\int_0^t \big( f^{T}(t-s) - f^{\alpha, \lambda}(t-s)\big)M^{*T}_s\mathrm{d}s = \int_0^t \big( F^{T}(t-s) - \frac{1}{2}F^{\alpha,\lambda}(t-s)\big)\mathrm{d}M^{*T}_s
$$
and using Burkholder-Davis-Gundy inequality we get ($C$ denotes here a positive constant that varies from line to line)
\begin{align*}
\mathbb{E}\Big[ \underset{t\in [0, 1]}{\sup} \Big(\int_0^t \big( F^{T}(t-s) - &\frac{1}{2} F^{\alpha,\lambda}(t-s)\big)\mathrm{d}M^{*T}_s \Big)^2\Big]\\
&\leq C \mathbb{E}\Big[\int_0^t \big( F^{T}(t-s) - \frac{1}{2}F^{\alpha,\lambda}(t-s)\big)^2\mathrm{d}X^{T}_s\Big]\\
&\leq C  \int_0^t \big( F^{T}(t-s) - \frac{1}{2}F^{\alpha, \lambda}(t-s)\big)^2 \frac{1-a_T}{\mu_T}\mathbb{E}[\lambda^T_{Ts}]\mathrm{d}s\\
&\leq C  \int_0^t \big( F^{T}(t-s) - \frac{1}{2}F^{\alpha,\lambda}(t-s)\big)^2 \mathrm{d}s.
\end{align*}
This converges to zero according to Lemma 4.3 in \cite{Jaisson:2016aa}. So we have proved that
$$
X_t = \int_0^t \frac{1}{2}F^{\alpha, \lambda}(t-s)( 1 + Z_s^2 )\mathrm{d}s+ \int_0^t f^{\alpha, \lambda}(t-s)\frac{1}{2\sqrt{\lambda\mu^*}}M_s\mathrm{d}s
$$
with
$$
Z_t^2 = \int_0^t k(t-s)\mathrm{d}P_s.
$$

\subsubsection{Regularity property }

We can write $X$ as
$$
X_t = \int_0^t \frac{1}{2} f^{\alpha, \lambda}(t-s)\big( s + \int_0^s Z_u^2\mathrm{d}u + M_s\big) \mathrm{d}s.
$$
Since $Z$ is continuous, $\int_0^tZ_s^2\mathrm{d}s$ is continuously differentiable. So using the same arguments as in Sections 4.3 and 4.4 in \cite{Jaisson:2016aa} replacing $s$ by $s + \int_0^s Z_u^2 \mathrm{d}u$, we obtain that almost surely, $X$ is differentiable with derivative $V$ satisfying
$$
V_t = \int_0^t \frac{1}{2}f^{\alpha, \lambda}(t-s) ( 1 + Z_s^2)\mathrm{d}s + \int_0^t \frac{1}{2}f^{\alpha, \lambda}(t-s) \frac{1}{\sqrt{\lambda\mu^*}}\mathrm{d}M_s.
$$
We get the stated result using Theorem V-3.9 in \cite{revuz2013continuous} which gives the existence of two independent Brownian motions $B^{(1)}$ and $B^{(2)}$ such that 
$$
M_t = \int_0^t \sqrt{V_s}\mathrm{d}B^{(1)}_s \text{ and }P_t = \int_0^t \sqrt{V_s}\mathrm{d}B^{(2)}_s.
$$
The regularity property of $V$ is also deduced using the same arguments as in Sections 4.3 and 4.4 in \cite{Jaisson:2016aa}.

\section*{Acknowledgments}

The authors are grateful to Jean-Philippe Bouchaud and Jim Gatheral for many interesting discussions.
The authors gratefully acknowledge financial support of the ERC Grant 679836 Staqamof and of the Chair {\it Analytics and Models for regulation}.

\bibliographystyle{abbrv}
\bibliography{bibliodjrarxivv2}

\begin{thebibliography}{10}

\bibitem{jaber2019sve}
E.~Abi~Jaber, C.~Cuchiero, M.~Larsson, and S.~Pulido.
\newblock Existence and stability for stochastic {V}olterra equations of
  convolution type with jumps.
\newblock {\em preprint}, 2019.

\bibitem{jaber2018multi}
E.~Abi~Jaber and O.~El~Euch.
\newblock Multi-factor approximation of rough volatility models.
\newblock {\em SIAM Journal on Financial Mathematics}, 10(2):309--349, 2019.

\bibitem{abi2017affine}
E.~Abi~Jaber, M.~Larsson, and S.~Pulido.
\newblock Affine {V}olterra processes.
\newblock {\em arXiv preprint arXiv:1708.08796}, 2017.

\bibitem{BACRY20132475}
E.~Bacry, S.~Delattre, M.~Hoffmann, and J.~Muzy.
\newblock Some limit theorems for {H}awkes processes and application to
  financial statistics.
\newblock {\em Stochastic Processes and their Applications}, 123(7):2475 --
  2499, 2013.

\bibitem{Bayer:2016aa}
C.~Bayer, P.~Friz, and J.~Gatheral.
\newblock Pricing under rough volatility.
\newblock {\em Quantitative Finance}, 16(6):887--904, 2016.

\bibitem{bennedsen2016decoupling}
M.~Bennedsen, A.~Lunde, and M.~S. Pakkanen.
\newblock Decoupling the short-and long-term behavior of stochastic volatility.
\newblock {\em arXiv preprint arXiv:1610.00332}, 2016.

\bibitem{bingham1989regular}
N.~H. Bingham, C.~M. Goldie, and J.~L. Teugels.
\newblock {\em Regular variation}, volume~27.
\newblock Cambridge university press, 1989.

\bibitem{blanc2017quadratic}
P.~Blanc, J.~Donier, and J.-P. Bouchaud.
\newblock Quadratic {H}awkes processes for financial prices.
\newblock {\em Quantitative Finance}, 17(2):171--188, 2017.

\bibitem{chicheportiche2014fine}
R.~Chicheportiche and J.-P. Bouchaud.
\newblock The fine-structure of volatility feedback.
\newblock {\em Physica A: Statistical Mechanics and its Applications},
  410:174--195, 2012.

\bibitem{da2018volatility}
J.~Da~Fonseca and W.~Zhang.
\newblock Volatility of volatility is (also) rough.
\newblock {\em Journal of Futures Markets}, 39:600--611, 2019.

\bibitem{el2018microstructural}
O.~El~Euch, M.~Fukasawa, and M.~Rosenbaum.
\newblock The microstructural foundations of leverage effect and rough
  volatility.
\newblock {\em Finance and Stochastics}, 22(2):241--280, 2018.

\bibitem{euch2018zumbach}
O.~El~Euch, J.~Gatheral, R.~Radoi{\v{c}}i{\'c}, and M.~Rosenbaum.
\newblock The {Z}umbach effect under rough {H}eston.
\newblock {\em arXiv preprint arXiv:1809.02098}, 2018.

\bibitem{el2020zumbach}
O.~El~Euch, J.~Gatheral, R.~Radoi{\v{c}}i{\'c}, and M.~Rosenbaum.
\newblock The zumbach effect under rough heston.
\newblock {\em Quantitative Finance}, 20(2):235--241, 2020.

\bibitem{el2018perfect}
O.~El~Euch and M.~Rosenbaum.
\newblock Perfect hedging in rough {H}eston models.
\newblock {\em The Annals of Applied Probability}, 28(6):3813--3856, 2018.

\bibitem{el2019characteristic}
O.~El~Euch and M.~Rosenbaum.
\newblock The characteristic function of rough {H}eston models.
\newblock {\em Mathematical Finance}, 29(1):3--38, 2019.

\bibitem{engle1982autoregressive}
R.~F. Engle.
\newblock Autoregressive conditional heteroscedasticity with estimates of the
  variance of united kingdom inflation.
\newblock {\em Econometrica: Journal of the Econometric Society},
  50(4):987--1007, 1982.

\bibitem{engle1986modelling}
R.~F. Engle and T.~Bollerslev.
\newblock Modelling the persistence of conditional variances.
\newblock {\em Econometric reviews}, 5(1):1--50, 1986.

\bibitem{filimonov2012quantifying}
V.~Filimonov and D.~Sornette.
\newblock Quantifying reflexivity in financial markets: Toward a prediction of
  flash crashes.
\newblock {\em Physical Review E}, 85(5):056108, 2012.

\bibitem{flandoli1995martingale}
F.~Flandoli and D.~Gatarek.
\newblock Martingale and stationary solutions for stochastic navier-stokes
  equations.
\newblock {\em Probability Theory and Related Fields}, 102(3):367--391, 1995.

\bibitem{gatheral2018volatility}
J.~Gatheral, T.~Jaisson, and M.~Rosenbaum.
\newblock Volatility is rough.
\newblock {\em Quantitative Finance}, 18(6):933--949, 2018.

\bibitem{glasserman2018buy}
P.~Glasserman and P.~He.
\newblock Buy rough, sell smooth.
\newblock {\em Working paper}, 2018.

\bibitem{hardiman2013critical}
S.~J. Hardiman, N.~Bercot, and J.-P. Bouchaud.
\newblock Critical reflexivity in financial markets: a hawkes process analysis.
\newblock {\em The European Physical Journal B}, 86(10):442, 2013.

\bibitem{jaber2019affine}
E.~A. Jaber, M.~Larsson, S.~Pulido, et~al.
\newblock Affine volterra processes.
\newblock {\em The Annals of Applied Probability}, 29(5):3155--3200, 2019.

\bibitem{jacod1975multivariate}
J.~Jacod.
\newblock Multivariate point processes: predictable projection, radon-nikodym
  derivatives, representation of martingales.
\newblock {\em Zeitschrift f{\"u}r Wahrscheinlichkeitstheorie und verwandte
  Gebiete}, 31(3):235--253, 1975.

\bibitem{jacod2013limit}
J.~Jacod and A.~Shiryaev.
\newblock {\em Limit theorems for stochastic processes}, volume 288.
\newblock Springer Science \& Business Media, 2013.

\bibitem{Jaisson:2015aa}
T.~Jaisson and M.~Rosenbaum.
\newblock Limit theorems for nearly unstable {H}awkes processes.
\newblock {\em The Annals of Applied Probability}, 25(2):600--631, 2015.

\bibitem{Jaisson:2016aa}
T.~Jaisson and M.~Rosenbaum.
\newblock Rough fractional diffusions as scaling limits of nearly unstable
  heavy tailed {H}awkes processes.
\newblock {\em The Annals of Applied Probability}, 26(5):2860--2882, 2016.

\bibitem{Jakubowski:1989aa}
A.~Jakubowski, J.~M{\'e}min, and G.~Pag{\`e}s.
\newblock Convergence en loi des suites d'int{\'e}grales stochastiques sur
  l'espace 1 de {S}korokhod.
\newblock {\em Probability Theory and Related Fields}, 81(1):111--137, 1989.

\bibitem{jusselin2018no}
P.~Jusselin and M.~Rosenbaum.
\newblock No-arbitrage implies power-law market impact and rough volatility.
\newblock {\em To appear in Mathematical Finance}, 2018.

\bibitem{Kurtz:1991aa}
T.~G. Kurtz and P.~Protter.
\newblock Weak limit theorems for stochastic integrals and stochastic
  differential equations.
\newblock {\em The Annals of Probability}, 19(3):1035--1070, 1991.

\bibitem{livieri2018rough}
G.~Livieri, S.~Mouti, A.~Pallavicini, and M.~Rosenbaum.
\newblock Rough volatility: evidence from option prices.
\newblock {\em IISE Transactions}, 50(9):767--776, 2018.

\bibitem{lynch2003market}
P.~E. Lynch and G.~Zumbach.
\newblock Market heterogeneities and the causal structure of volatility.
\newblock {\em Quantitative Finance}, 3(4):320--331, 2003.

\bibitem{nelson1990arch}
D.~B. Nelson.
\newblock {ARCH} models as diffusion approximations.
\newblock {\em Journal of econometrics}, 45(1-2):7--38, 1990.

\bibitem{ogata1981lewis}
Y.~Ogata.
\newblock On {L}ewis' simulation method for point processes.
\newblock {\em IEEE Transactions on Information Theory}, 27(1):23--31, 1981.

\bibitem{Protter:2005aa}
P.~E. Protter.
\newblock {\em Stochastic differential equations}.
\newblock Springer, 2005.

\bibitem{revuz2013continuous}
D.~Revuz and M.~Yor.
\newblock {\em Continuous martingales and Brownian motion}, volume 293.
\newblock Springer Science \& Business Media, 2013.

\bibitem{sentana1995quadratic}
E.~Sentana.
\newblock Quadratic {ARCH} models.
\newblock {\em The Review of Economic Studies}, 62(4):639--661, 1995.

\bibitem{veraar2012stochastic}
M.~Veraar.
\newblock The stochastic {F}ubini theorem revisited.
\newblock {\em Stochastics An International Journal of Probability and
  Stochastic Processes}, 84(4):543--551, 2012.

\bibitem{zhang2010stochastic}
X.~Zhang.
\newblock Stochastic {V}olterra equations in {B}anach spaces and stochastic
  partial differential equation.
\newblock {\em Journal of Functional Analysis}, 258(4):1361--1425, 2010.

\bibitem{zumbach2009time}
G.~Zumbach.
\newblock Time reversal invariance in finance.
\newblock {\em Quantitative Finance}, 9(5):505--515, 2009.

\bibitem{zumbach2010volatility}
G.~Zumbach.
\newblock Volatility conditional on price trends.
\newblock {\em Quantitative Finance}, 10(4):431--442, 2010.

\end{thebibliography}

\end{document}